%% file: ms.tex
\documentclass[a4paper,USenglish,cleveref, autoref, thm-restate]{lipics-v2021}

\usepackage{graphicx}
\graphicspath{ {./img/} }

\usepackage{subcaption}
\usepackage[colorinlistoftodos, prependcaption, textsize=scriptsize]{todonotes}

\usepackage[ruled,vlined,linesnumbered]{algorithm2e}
\DontPrintSemicolon
\SetKwInOut{Input}{\footnotesize Input}
\SetKwInOut{Output}{\footnotesize Output}
\SetKw{KwDownTo}{downto}
\SetKw{KwAnd}{and}
\SetKw{KwOr}{or}
\SetKw{KwNot}{not}
\SetKw{KwTrue}{true}
\SetKw{KwFalse}{false}
\SetKw{KwDownto}{downto}
\newcommand{\IO}[2]{\Input{\footnotesize #1} \Output{\footnotesize #2}\BlankLine\BlankLine }

\SetKwProg{Fn}{Proc}{}{}

\include{notation}

\title{A Fast and Small Subsampled R-index}

\author{Dustin Cobas}{CeBiB --- Center for Biotechnology and Bioengineering, Chile \and Dept. of Computer Science, University of Chile, Chile}{dustin.cobas@gmail.com}{https://orcid.org/0000-0001-6081-694X}{ANID/Scholarship Program/DOCTORADO BECAS CHILE/2020-21200906, Chile.}

\author{Travis Gagie}{CeBiB --- Center for Biotechnology and Bioengineering, Chile \and Dalhousie University, Canada}{travis.gagie@gmail.com}{https://orcid.org/0000-0003-3689-327X}{Funded by NSERC Discovery Grant RGPIN-07185-2020.}

\author{Gonzalo Navarro}{CeBiB --- Center for Biotechnology and Bioengineering, Chile \and Dept. of Computer Science, University of Chile, Chile \and \url{http://www.dcc.uchile.cl/gnavarro} }{gnavarro@dcc.uchile.cl}{https://orcid.org/0000-0002-2286-741X}{Fondecyt grant 1-200038, ANID, Chile.}

\authorrunning{D. Cobas, T. Gagie, and G. Navarro} 

\Copyright{Dustin Cobas and Travis Gagie and Gonzalo Navarro} 



\category{} 

\relatedversion{} 

\acknowledgements{Funded with Basal Funds FB0001, ANID, Chile.}

\nolinenumbers 

\hideLIPIcs  

\EventLongTitle{32nd Annual Symposium on Combinatorial Pattern Matching (CPM 2021)}
\EventShortTitle{CPM 2021}
\EventAcronym{CPM}
\EventYear{2021}

\begin{document}

\maketitle

\begin{abstract}
The $\RIdx$ (Gagie et al., JACM 2020) represented a breakthrough in compressed indexing of repetitive text collections, outperforming its alternatives by orders of magnitude. Its space usage, $\Oh(r)$ where $r$ is the number of runs in the Burrows-Wheeler Transform of the text, is however larger than Lempel-Ziv and grammar-based indexes, and makes it uninteresting in various real-life scenarios of milder repetitiveness. In this paper we introduce the $\SRIdx$, a variant that limits the space to $\Oh(\min(r,n/s))$ for a text of length $n$ and a given parameter $s$, at the expense of multiplying by $s$ the time per occurrence reported. The $\SRIdx$ is obtained by carefully subsampling the text positions indexed by the $\RIdx$, in a way that we prove is still able to support pattern matching with guaranteed performance. Our experiments demonstrate that the $\SRIdx$ sharply outperforms virtually every other compressed index on repetitive texts, both in time and space, even matching the performance of the $\RIdx$ while using 1.5--3.0 times less space.  Only some Lempel-Ziv-based indexes achieve better compression than the $\SRIdx$, using about half the space, but they are an order of magnitude slower.
\end{abstract}

\include{body}



\typeout{}
\bibliography{bibliography}


\end{document}

%% file: notation.tex
\newcommand\Oh{{\mathcal{O}}}

\newcommand\vTColl{\mathcal{T}}	\newcommand\vT{\vTColl}
\newcommand\vsTColl{\mathit{n}}	\newcommand\vn{\vsTColl}
\newcommand\vPat{\mathit{P}}	\newcommand\vP{\vPat}
\newcommand\vsPat{\mathit{m}}	\newcommand\vm{\vsPat}

\newcommand\occ{\mathit{occ}}

\newcommand\vSPosPat{\mathit{sp}}     \newcommand\vsp{\vSPosPat}
\newcommand\vEPosPat{\mathit{ep}}     \newcommand\vep{\vEPosPat}

\newcommand\vnBWTRuns{\mathit{r}}	\newcommand\vr{\vnBWTRuns}

\newcommand\vSamp{\mathit{t}'}
\newcommand\vt{\mathit{t}}
\newcommand\vSampFac{s}			\newcommand\vs{\vSampFac}

\newcommand\SA{\mathsf{SA}}
\newcommand\CSA{\mathsf{CSA}}
\newcommand\FMIdx{\textsf{fm-index}}
\newcommand\RLCSA{\mathsf{rlcsa}}
\newcommand\RLFMIdx{\textsf{rlfm-index}}
\newcommand\RIdx{\textsf{\emph{r}-index}}
\newcommand\SRIdx{\textsf{\emph{sr}-index}}
\newcommand\LZIdx{\textsf{lz-index}}
\newcommand\GIdx{\textsf{\emph{g}-index}}
\newcommand\HybridIdx{\textsf{hyb-index}}
\newcommand\LZEndIdx{\textsf{lze-index}}

\newcommand\BWT{\mathsf{BWT}}
\newcommand\LF{\mathsf{LF}}

\DeclareMathOperator{\ftSearch}{t_{search}(\vsPat)}
\DeclareMathOperator{\ftLookup}{t_{lookup}(\vsTColl)}

\newcommand\nameDS{\texttt}
\newcommand\nameOp{\texttt}

\newcommand\rank{\nameOp{rank}}
\newcommand\select{\nameOp{select}}
\newcommand\pred{\nameOp{pred}}

\newcommand\Start{\nameDS{Start}}
\newcommand\Letter{\nameDS{Letter}}
\newcommand\First{\nameDS{First}}
\newcommand\FirstToRun{\nameDS{FirstToRun}}
\newcommand\Samples{\nameDS{Samples}}
\newcommand\Valid{\nameDS{Valid}}
\newcommand\ValidArea{\nameDS{ValidArea}}
\newcommand\Removed{\nameDS{Removed}}

%% file: body.tex
\section{Introduction} \label{sec:intro}

The rapid surge of massive repetitive text collections, like genome and sequence read sets and versioned document and software repositories, has raised the interest in text indexing techniques that exploit repetitiveness to obtain orders-of-magnitude space reductions, while supporting pattern matching directly on the compressed text representations \cite{GNencyc18,Nav20}.

Traditional compressed indexes rely on statistical compression \cite{NavarroM07:CFT}, but this is ineffective to capture repetitiveness \cite{KreftN13:CIR}. A new wave of repetitiveness-aware indexes \cite{Nav20} build on other compression mechanisms like Lempel-Ziv \cite{LZ76} or grammar compression \cite{KY00}. A particularly useful index of this kind is the $\RLFMIdx$ \cite{MakinenN05:SSA,MakinenNSV10:SRH}, because it emulates the classical suffix array \cite{ManberM93:SAN} and this simplifies translating suffix-array based algorithms to run on it \cite{MBCT15}.

The $\RLFMIdx$ represents the 
Burrows-Wheeler Transform ($\BWT$) \cite{BurrowsW94:BSL} of the text in run-length compressed form, because the number $r$ of maximal equal-letter runs in the $\BWT$ is known to be small on repetitive texts. A problem with the $\RLFMIdx$ is that, although it can count the number of occurrences of a pattern using $\Oh(\vr)$ space, it needs to sample the text at every $\vs$th position, for a parameter $\vs$, in order to locate each of those occurrences in time proportional to $\vs$. The $\Oh(\vn/\vs)$ additional space incurred on a text of length $\vn$ ruins the compression on very repetitive collections, where $\vr \ll \vn$.  
The recent $\RIdx$ \cite{GagieNP20:FFS} closed the long-standing problem of efficiently locating the occurrences within $\Oh(\vnBWTRuns)$ space, offering pattern matching time orders of magnitude faster than previous repetitiveness-aware indexes.

In terms of space, however, the $\RIdx$ is considerably larger
than Lempel-Ziv based indexes of size $\Oh(z)$, where $z$ is the number of phrases in the Lempel-Ziv parse. Gagie et al.~\cite{GagieNP20:FFS} show that, on extremely repetitive text collections where $n/r = 500$--$10{,}000$, $r$ is around $3z$ and the $\RIdx$ size is $0.06$--$0.2$ bits per symbol (bps), about twice that of the $\LZIdx$ \cite{KreftN13:CIR}, a baseline Lempel-Ziv based index. However, $r$ degrades faster than $z$ as repetitiveness drops: in an experiment on bacterial genomes in the same article, where $n/r \approx 100$, the $\RIdx$ space approaches $0.9$ bps, $4$ times that of the $\LZIdx$; $r$ also approaches $4z$. 
Experiments on other datasets show that the $\RIdx$ tends to be considerably larger \cite{NS19,CNP21,DN21,BCGHMNR21}.\footnote{The $n/r$ measurements in the article \cite{BCGHMNR21} are not correct.}
Indeed, in some realistic cases $n/r$ can be over 1{,}500, but in most cases it is well below: 40--160 on versioned software and document collections and fully assembled human chromosomes, 7.5--50 on virus and bacterial genomes (with $r$ in the range $4z$--$7z$), and 4--9 on sequencing reads; see Section~\ref{sec:result}. 
An $\RIdx$ on such a small $n/r$ ratio easily becomes larger than the plain sequence data.

In this paper we tackle the problem of the (relatively) large space usage of the $\RIdx$. This index manages to locate the pattern occurrences by sampling $\vr$ text positions (corresponding to the ends of $\BWT$ runs). We show that one can remove some carefully chosen samples so that, given a parameter $s$, the index stores only $\Oh(\min(r,n/s))$ samples while its locating machinery can still be used to guarantee that every pattern occurrence is located within $\Oh(s)$ steps. We call the resulting index the {\em subsampled $\RIdx$}, or $\SRIdx$. The worst-case time to locate the $\occ$ occurrences of a pattern of length $\vm$ on an alphabet of size $\sigma$ then rises from $\Oh((\vm+\occ)\log(\sigma+n/r))$ in the implemented $\RIdx$ to $\Oh((\vm+\vs \cdot \occ)\log(\sigma+n/r))$ in the $\SRIdx$, which matches the search cost of the $\RLFMIdx$. 

The $\SRIdx$ can then be seen as a hybrid between the $\RIdx$ (matching it when $s=1$) and the $\RLFMIdx$ (obtaining its time with less space; the spaces become similar when repetitiveness drops). In practice, however, the $\SRIdx$ performs much better than both on repetitive texts, sharply dominating the $\RLFMIdx$, the best grammar-based index \cite{CNP21}, and in most cases the $\LZIdx$, both in space and time.
The $\SRIdx$ can also get as fast as the $\RIdx$ while using $1.5$--$4.0$ times less space. 
Its only remaining competitor is a hybrid between a Lempel-Ziv based and a statistical index \cite{FKP18}. This index can use up to half the space of the $\SRIdx$, but it is an order of magnitude slower. Overall, the $\SRIdx$ stays orders of magnitude faster than all the alternatives while using practical amounts of space in a wide range of repetitiveness scenarios.


\section{Background} \label{sec:background}


The {\em suffix array} \cite{ManberM93:SAN} $\SA[1..\vsTColl]$ of a string $\vTColl[1..\vsTColl]$ over alphabet $[1..\sigma]$ is a permutation of the starting positions of all the suffixes of $\vTColl$ in lexicographic order, $\vTColl[\SA[i].. \vsTColl] < \vTColl[\SA[i + 1].. \vsTColl]$ for all $1 \leq i < \vsTColl$. The suffix array
can be binary searched in time $\Oh(m\log n)$ to obtain the range $\SA[\vSPosPat.. \vEPosPat]$ of all the suffixes prefixed by a search pattern $\vPat[1..\vm]$ (which then occurs $\occ = \vEPosPat - \vSPosPat + 1$ times in $\vT$). Once they are {\em counted} (i.e., their suffix array range is determined), those occurrences are {\em located} in time $\Oh(\occ)$ by simply listing their starting positions, $\SA[\vSPosPat],\ldots, \SA[\vEPosPat]$.
The suffix array can then be stored in $\vsTColl \lceil \lg \vsTColl \rceil$ bits (plus the $n\lceil\lg \sigma\rceil$ bits to store $\vT$) and searches for $\vP$ in $\vT$ in total time $\Oh(m\log n + \occ)$.

\emph{Compressed suffix arrays} ($\CSA$s) \cite{NavarroM07:CFT} are space-efficient representations of both the suffix array ($\SA$) and the text ($\vT$).
They can find the interval $\SA[\vSPosPat..\vEPosPat]$ corresponding to $\vPat[1..\vsPat]$ in time $\ftSearch$ and access any cell $\SA[j]$ in time $\ftLookup$, so they can be used to search for $\vPat$ in time $\Oh(\ftSearch + \occ \ftLookup)$. 
Most $\CSA$s need to store sampled $\SA$ values to compute any $\SA[j]$ in order to support the locate operation, inducing the tradeoff of using $\Oh((\vsTColl / s)\log n)$ extra bits to obtain time $\ftLookup$ proportional to a parameter $s$.

The \emph{Burrows-Wheeler Transform} \cite{BurrowsW94:BSL} of $\vTColl$ is a permutation $\BWT[1..\vsTColl]$ of $\vTColl[1..\vsTColl]$ defined as $\BWT[i]=\vTColl[\SA[i]-1]$ (and $\vTColl[\vsTColl]$ if $\SA[i]=1$), which boosts the compressibility of $\vTColl$. The $\FMIdx$ \cite{FerraginaM05:ICT, FerraginaMMN07:CRS} is a $\CSA$ that represents $\SA$ and $\vT$ within the \emph{statistical entropy} of $\vT$, by exploiting the connection between the $\BWT$ and $\SA$.
For counting, the $\FMIdx$ resorts to \emph{backward search}, which successively finds the suffix array ranges $\SA[\vsp_i..\vep_i]$ of $\vP[i..\vm]$, for $i=\vm$ to $1$, starting from $\SA[\vsp_{\vm+1}..\vep_{\vm+1}] = [1..\vn]$ and then
\begin{eqnarray*}
\vsp_i &=& C[c] + \rank_c(\BWT,\vsp_{i+1}-1)+1, \\
\vep_i &=& C[c] + \rank_c(\BWT,\vep_{i+1}),
\end{eqnarray*}
where $c=\vP[i]$, $C[c]$ is the number of occurrences of symbols smaller than $c$ in $\vT$, and $\rank_c(\BWT,j)$ is the number of times $c$ occurs in $\BWT[1..j]$. Thus, $[\vsp,\vep]=[\vsp_1,\vep_1]$ if $\vsp_i \le \vep_i$ holds for all $1 \le i \le m$.

For locating the occurrences $\SA[\vSPosPat],\ldots, \SA[\vEPosPat]$, the $\FMIdx$ uses $\SA$ sampling as described: it stores sampled values of $\SA$ at regularly spaced text positions, say multiples of $s$. This is done via the so-called {\em LF-steps}: The $\BWT$ allows one to efficiently compute, given $j$ such that $\SA[j]=i$, the value $j'$ such that $\SA[j']=i-1$, called $j'=\LF(j)$. The formula is
\[ \LF(i) ~=~ C[c] + \rank_c(\BWT,i),
\]
where $c=\BWT[i]$.
Note that the LF-steps virtually traverse the text backwards. By marking with $1$s in a bitvector $B[1..\vn]$ 
the positions $j^*$ such that $\SA[j^*]$ is a multiple of $s$, we can start
from any $j$ and, in $k < s$ LF-steps, find some sampled position $j^*=\LF^k(j)$
where $B[j^*]=1$. By storing those values $\SA[j^*]$ explicitly, we have $\SA[j]=\SA[j^*]+k$.

By implementing $\BWT$ with a wavelet tree, for example, access and $\rank_c$ on $\BWT$ can be supported in time $\Oh(\log\sigma)$, and the $\FMIdx$ searches in time $\Oh((m+s\cdot occ)\log\sigma)$ \cite{FerraginaMMN07:CRS}.

Since the statistical entropy is insensitive to repetitiveness \cite{KreftN13:CIR}, however, the $\FMIdx$ is not adequate for repetitive datasets.
The \emph{Run-Length FM-index}, $\RLFMIdx$ (and its variant $\RLCSA$) \cite{MakinenN05:SSA, MakinenNSV10:SRH}, is a modification of the $\FMIdx$ aimed at repetitive texts. Say that the $\BWT[1..\vsTColl]$ is formed by $\vnBWTRuns$ maximal runs of equal symbols, then $\vnBWTRuns$ is relatively small in repetitive collections 
(in particular, $\vnBWTRuns = \Oh(z\log^2 \vsTColl)$, where $z$ is the number 
of phrases of the Lempel-Ziv parse of $\vTColl$ \cite{KK19}). The
$\RLFMIdx$ supports counting within $\Oh(\vnBWTRuns\log n)$ bits, by implementing the backward search over alternative data structures. In particular, it marks in a bitvector $\Start[1..\vn]$ with $1$s the positions $j$ starting $\BWT$ runs, that is, where $j=1$ or $\BWT[j] \not= \BWT[j-1]$. The first letter of each run is collected in an array $\Letter[1..\vr]$. Since $\Start$ has only $r$ $1$s, it can be represented within $\vr\lg(\vn/\vr)+\Oh(\vr)$ bits. Within this space, one can access any bit $\Start[j]$ and support operation $\rank_1(\Start,j)$, which counts the number of $1$s in $\Start[1..j]$, in time $\Oh(\log(\vn/\vr))$ \cite{OS07}. Therefore, we simulate $\BWT[j]=\Letter[\rank_1(\Start,j)]$ in $\Oh(\vr\log n)$ bits. The backward search formula can be efficiently simulated as well, leading to $\Oh((m+s\cdot occ)\log(\sigma+n/r))$ search time.
However, the $\RLFMIdx$ still uses $\SA$ samples to locate, and when $\vnBWTRuns \ll \vsTColl$ (i.e., on repetitive texts), the $\Oh((\vsTColl / s)\log \vn)$ added bits ruin the $\Oh(\vr\log\vn)$-bit space ($s$ is typically $\Oh(\log n)$ or close).

The $\RIdx$ \cite{GagieNP20:FFS} closed the long-standing problem of efficiently locating the occurrences of a pattern in a text using $\Oh(\vnBWTRuns\log n)$-bit space.
The experiments showed that the $\RIdx$ outperforms all the other implemented indexes by orders of magnitude in space or in time to locate pattern occurrences on highly repetitive datasets. However, other experiments on more typical repetitiveness scenarios \cite{NS19,CNP21,DN21,BCGHMNR21} showed that the space of the $\RIdx$ degrades very quickly as repetitiveness decreases. For example, a grammar-based index (which can be of size $g = \Oh(z\log(\vsTColl/z))$) is usually slower but significantly smaller \cite{CNP21}, and an even slower Lempel-Ziv based index of size $O(z)$ \cite{KreftN13:CIR} is even smaller. Some later proposals \cite{NT20} further speed up the $\RIdx$ by increasing the constant accompanying the $\Oh(\vr\log\vn)$-bit space. The unmatched time performance of the $\RIdx$ comes then with a very high price in space on all but the most highly repetitive text collections, which makes it of little use in many relevant application scenarios. This is the problem we address in this paper.

\section{The $\RIdx$ Sampling Mechanism} \label{sec:rindex}

Gagie et al.~\cite{GagieNP20:FFS} provide an $\Oh(\vnBWTRuns\log\vn)$-bits data
structure that not only finds the range $\SA[\vSPosPat.. \vEPosPat]$ of the
occurrences of $\vPat$ in $\vTColl$, but also gives the value
$\SA[\vEPosPat]$, that is, the text position of the last occurrence in the 
range. They then provide a second $\Oh(\vnBWTRuns\log n)$-bits data structure that,
given $\SA[j]$, efficiently finds $\SA[j-1]$. This suffices to efficiently find all the occurrences of $\vPat$, in time $\Oh((m+occ)\log\log(\sigma+n/r))$ in their theoretical version.

In addition to the theoretical design, Gagie et al.\ and Boucher et al.~\cite{GagieNP20:FFS,BGKLMM19} provided a carefully engineered $\RIdx$ implementation. The counting data structures (which  
find the range $\SA[\vSPosPat.. \vEPosPat]$) require, for any small constant 
$\epsilon>0$, $\vr \cdot ((1+\epsilon)\lg(\vn/\vr)+\lg\sigma+\Oh(1))$ bits
(largely dominated by the described arrays $\Start$ and $\Letter$), whereas the 
locating data structures (which obtain $\SA[\vEPosPat]$, and $\SA[j-1]$ given $\SA[j]$), require
$\vr \cdot (2\lg \vn + \Oh(1))$ further bits. The locating structures are then significantly
heavier in practice, especially when $\vn/\vr$ is not that large. 
Together, the structures use $\vr \cdot ((1+\epsilon)\lg(\vn/\vr)+2\lg \vn + \lg\sigma+\Oh(1))$ bits of space and perform backward search steps and LF-steps in time $\Oh(\log(\sigma+\vn/\vr))$, so they search for $\vP$ in time
$\Oh((\vm+\occ)\log(\sigma+\vn/\vr))$.

For conciseness we do not describe the counting data structures of the $\RIdx$, which are the same of the $\RLFMIdx$ and which we do not modify in our index. The $\RIdx$ locating structures, which we do modify, are formed by the following components:

\begin{description}
\item[{$\First[1..\vsTColl]$}:] a bitvector marking with $1$s the {\em text} positions of the letters that are the first in a $\BWT$ run. That is, if $j=1$ or $\BWT[j] \not= \BWT[j-1]$, then $\First[\SA[j]-1]=1$. 
Since $\First$ has only $\vr$ $1$s, it is represented in compressed form using $\vr\lg(\vn/\vr)+\Oh(\vr)$ bits, while supporting $\rank_1$ in time $\Oh(\log(\vn/\vr))$ and, in $\Oh(1)$ time, the operation
$\select_1(\First,j)$ (the position of the $j$th $1$ in $\First$) \cite{OS07}. This
allows one find the rightmost $1$ up to position $i$, $\pred(\First,i)=\select_1(\First,\rank_1(\First,i))$.
\item[{$\FirstToRun[1..\vnBWTRuns]$}:] a vector of integers 
(using $\vr\lceil \lg \vnBWTRuns \rceil$ bits) mapping each letter marked in $\First$ to the $\BWT$ run where it lies. That is, if the
$p$th $\BWT$ run starts at $\BWT[j]$, and $\First[i]=1$ for
$i=\SA[j]-1$, then $\FirstToRun[\rank_1(\First,i)]=p$.
\item[{$\Samples[1..\vnBWTRuns]$}:] a vector of $\lceil \lg \vn\rceil$-bit integers storing samples
of $\SA$, so that $\Samples[p]$ is the text position $\SA[j]-1$ corresponding to the last letter $\BWT[j]$ in the $p$th $\BWT$ run. 
\end{description}

These structures are used in the following way in the $\RIdx$ implementation \cite{GagieNP20:FFS}:

\begin{description}
\item[Problem 1:] When computing the ranges $\SA[\vSPosPat.. \vEPosPat]$ along the backward search, we must also produce the value $\SA[\vEPosPat]$. They actually compute all the values $\SA[\vep_i]$. This is stored for $\SA[\vep_{\vm+1}]=\SA[n]$ and then, if $\BWT[\vep_{i+1}]=\vP[i]$, we know that $\vep_i=\LF(\vep_{i+1})$ and thus $\SA[\vep_i] = \SA[\vep_{i+1}]-1$. Otherwise, $\vep_i=\LF(j)$ and $\SA[\vep_i]=\SA[j]-1$, where $j \in [\vsp_{i+1}..\vep_{i+1}]$ is the largest position with $\BWT[j]=\vP[i]$. The position $j$ is efficiently found with their counting data structures, and the remaining problem is how to compute $\SA[j]$. Since $j$ must be an end of run, however, this is simply computed as $\Samples[p]+1$, where 
$p=\rank_1(\Start,j)$ is the run where $j$ belongs.
\item [Problem 2:] When locating we must find $\SA[j-1]$ from $i=\SA[j]-1$. There are two cases: \begin{itemize}
    \item $j-1$ ends a $\BWT$ run, that is, $\Start[j]=1$, and then $\SA[j-1]=\Samples[p-1]+1$, where $p$ is as in Problem 1;
    \item $j-1$ is in the same $\BWT$ run of $j$, in which case they compute $\SA[j-1]=\phi(i)$, where
\begin{equation} \label{eq:rindex}
\phi(i) = \Samples[\FirstToRun[\rank_1(\First,i)]-1]+1+(i-\pred(\First,i)).
\end{equation}
\end{itemize}
\end{description}

\begin{figure}[t]
\begin{center}
\includegraphics[width=\textwidth]{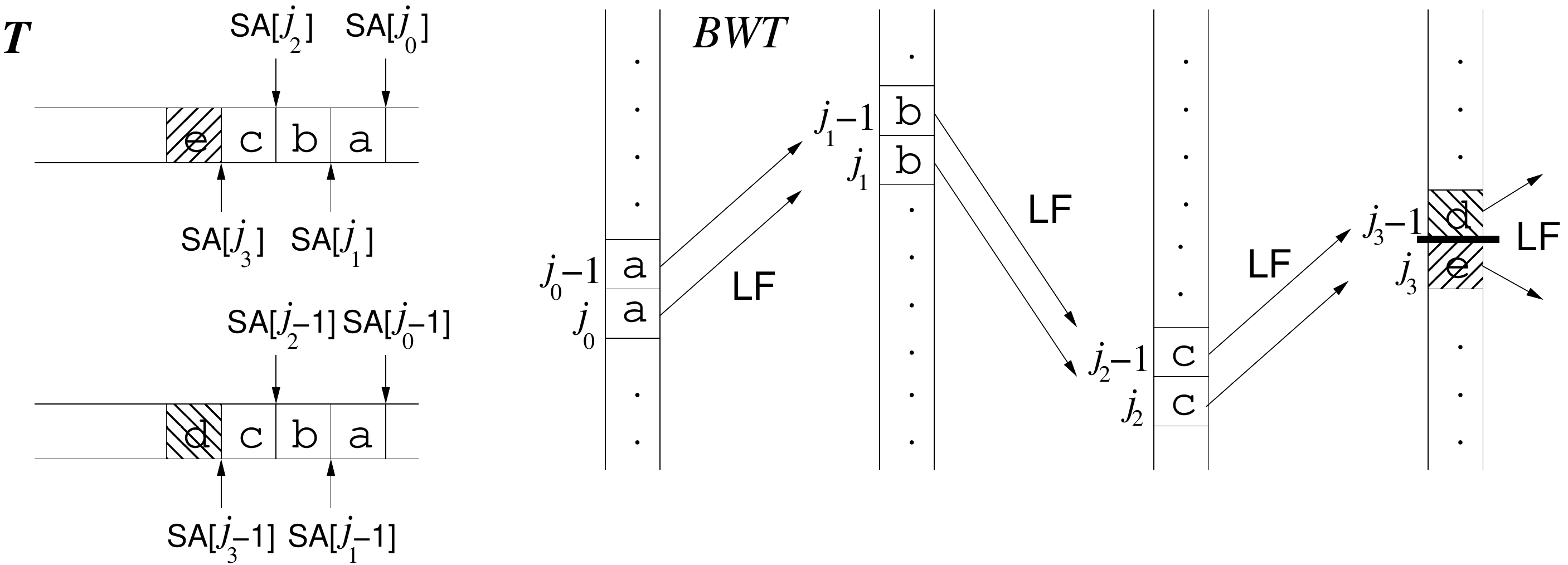}
\end{center}
\caption{Schematic example of the sampling mechanism of the $\RIdx$. There is
a run border between $j_3-1$ and $j_3$.}
\label{fig:figura}
\end{figure}

This formula works because, when $j$ and $j-1$ are in the same $\BWT$ run,
it holds that $\LF(j-1)=\LF(j)-1$ \cite{FerraginaM05:ICT}. 
Figure~\ref{fig:figura} explains why this property makes the formula work. 
Consider two $\BWT$
positions, $j=j_0$ and $j' = j-1 = j_0-1$, that belong to the same run. The 
$\LF$ formula will map them to consecutive positions, $j_1$ and $j_1'=j_1-1$.
If $j_1$ and $j_1-1$ still belong to the same run, $\LF$ will map them to
consecutive positions again, $j_2$ and $j_2'=j_2-1$, and once again, $j_3$
and $j_3'=j_3-1$. Say that $j_3$ and $j_3-1$ do not belong to the same run.
This means that $j_3-1$ ends a run (and thus it is stored in $\Samples$)
and $j_3$ starts a run (and thus $\SA[j_3]-1$ is marked in $\First$).
To the left of the $\BWT$ positions we show the areas of $\vT$ virtually traversed as
we perform consecutive LF-steps. Therefore, if we know $i=\SA[j]-1=\SA[j_0]-1$,
the nearest $1$ in $\First$ to the left is at $\pred(\First,i)=\SA[j_3]-1$ 
(where there is an \texttt{e} in $\vT$) and $p=\FirstToRun[\rank(i)]$ is 
the number of the $\BWT$ run that starts at $j_3$. If we subtract $1$, we have the
$\BWT$ run ending at $j_3-1$, and then
$\Samples[p-1]$ is the position preceding
$\SA[j_3-1]$ (where there is a \texttt{d} in $\vT$). We add $1+(i-\pred(\First,i))=4$
to obtain $\SA[j_0-1]=\SA[j-1]$.

These components make up, effectively, a sampling mechanism of $\Oh(\vr\log n)$ bits (i.e., sampling the end of runs),
instead of the traditional one of  $\Oh((n/s)\log n)$ bits (i.e., sampling every $s$th text position).

\section{Our Subsampled $\RIdx$}

Despite its good performance on highly repetitive texts, the sampling mechanism introduced by the $\RIdx$ is excessive in areas where the $\BWT$ runs are short, because those induce oversampled ranges on the text. In this section we describe an $\RIdx$ variant we dub {\em subsampled $\RIdx$}, or {\em $\SRIdx$}, which 
can be seen as a hybrid between the $\RIdx$ and the $\RLFMIdx$.
The $\SRIdx$ samples the text at end of runs (like the $\RIdx$), but in oversampled areas it removes some samples to ensure that no three consecutive samples lie within distance $s$ (roughly as in the 
$\RLFMIdx$).
It then handles text areas with denser and sparser sampling in different ways.

\subsection{Subsampling}

The $\SRIdx$ subsampling process removes $\RIdx$ samples in oversampled areas. Concretely, let $\vSamp_1 < \cdots < \vSamp_{\vnBWTRuns}$ be the text positions of the last letters in $\BWT$ runs, that is, the sorted values in array $\Samples$. For any $1 < i < \vr$, we remove the sample $\vSamp_{i}$ if $\vSamp_{i+1} - \vSamp_{i-1} \le \vSampFac$, where $\vs$ is a parameter. This condition is tested and applied sequentially for $i = 2,\ldots,\vr-1$ (that is, if we 
removed $\vSamp_2$ because $\vSamp_3-\vSamp_1 \le \vs$, then we next remove
$\vSamp_3$ if $\vSamp_4-\vSamp_1 \le \vs$; otherwise we remove
$\vSamp_3$ if $\vSamp_4-\vSamp_2 \le \vs$). Let us call $\vt_1,\vt_2,\ldots$ the sequence of the remaining samples.

The arrays $\First$, $\FirstToRun$, and $\Samples$ are built on the 
samples $\vt_i$ only. That is, if we remove the sample $\Samples[p]=\vSamp$, we also remove the $1$ in $\First$ corresponding to the first letter of the $(p+1)$th $\BWT$ run, which is the one Eq.~(\ref{eq:rindex}) would have handled with $\Samples[p]$. We also remove the corresponding entry of $\FirstToRun$. Note that, if $j$ is the first position of the $(p+1)$th run and $j-1$ the last of the $p$th run, then if we remove $\Samples[p]=\SA[j-1]-1$, we remove the corresponding $1$ at position $\SA[j]-1$ in $\First$.

It is not hard to see that subsampling avoids the excessive space usage when $r$ is not small enough, reducing it from $\Oh(\vr)$ to $\Oh(\min(\vr,\vn/\vs))$ entries for the locating structures.

\begin{lemma} \label{lem:space}
The subsampled structures {\em $\First$}, {\em $\FirstToRun$}, and {\em $\Samples$} use 
$\min(\vr,2\lceil \vn/(\vs+1)\rceil) \cdot (2\lg n + \Oh(1))$ bits of space.
\end{lemma}
\begin{proof}
This is their same space as in the implemented $\RIdx$, with the number of samples reduced from $\vr$ to $\min(\vr,2\lceil \vn/(\vs+1)\rceil)$. We start with $\vr$ samples and remove some, so there are at most $r$. By construction, any remaining sample $\vt_i$ satisfies
$\vt_{i+1} - \vt_{i-1} > \vSampFac$, so if we cut the text into blocks of length $s+1$, no block can contain more than $2$ samples.
\end{proof}

Our index adds the following small structure on top of the above ones, so as to mark the removed samples:

\begin{description}
\item[{$\Removed[1..r]$}:] A bitvector telling which of the original samples
have been removed, that is, $\Removed[p]=1$ iff the sample at the end of the $p$th
$\BWT$ run was removed. We can compute any $\rank_1(\Removed,p)$ in constant time using $r+o(r)$ bits \cite{Cla96}.
\end{description}

It is easy to see that, once the $\RIdx$ structures are built, the $\SRIdx$ subsampling, as well as building and updating the associated structures, are lightweight tasks, easily carried out in $\Oh(r)$ space and $\Oh(r\log r)$ time. It is also possible to build the subsampled structures directly without building the full $\SRIdx$ sampling first, in $\Oh(\vn\log(\sigma+\vn/\vr))$ time: we simulate a backward text traversal using $\LF$-steps, so that we can build bitvector $\Removed$. A second similar traversal fills the $1$s in $\First$ and the entries in $\FirstToRun$ and $\Samples$ for the runs whose sample was not removed.

\subsection{Solving Problem 1}

For Problem 1, we must compute $\SA[j]$, where $j$ is the end of the $p$th run, with $p=\rank_1(\Start,j)$. This position is sampled in the $\RIdx$, where the problem is thus trivial: $\SA[j]=\Samples[p]+1$. However, in the $\SRIdx$ it might be that $\Removed[p]=1$, which means that the subsampling process removed $\SA[j]$. In this case, we compute $j_k=\LF^k(j)$ for $k=1,2,\ldots$ until finding a sampled value $\SA[j_k]$ (i.e., $j_k=n$ or $\Start[j_k+1]=1$) that is not removed (i.e., $q = \rank_1(\Start,j_k)$ and $\Removed[q]=0$). We then compute $q' = q - \rank_1(\Removed,q)$, and $\SA[j] = \Samples[q']+k+1$.

The next lemma shows that we find a nonremoved sample for some $k < s$.

\begin{lemma} \label{lem:distance-s}
If there is a removed sample $\vSamp_j$ such that $\vt_i < \vSamp_j < \vt_{i+1}$, then $\vt_{i+1}-\vt_i \le s$.
\end{lemma}
\begin{proof}
Since our subsampling process removes samples left to right, by the time we removed $\vSamp_j$, the current sample $\vt_i$ was already the nearest remaining sample to the left of $\vSamp_j$. If the sample following $\vSamp_j$ was the current $\vt_{i+1}$, then we removed $\vSamp_j$ because $\vt_{i+1}-\vt_i \le s$, and we are done. Otherwise, there were other samples to the right of $\vSamp_j$, say $\vSamp_{j+1}, \vSamp_{j+2},\ldots, \vSamp_{j+k}$, that were consecutively removed until reaching the current sample $\vt_{i+1}$. We removed $\vSamp_j$ because $\vSamp_{j+1} - \vt_i \le s$. Then, for $1 \le l < k$, we removed $\vSamp_{j+l}$  (after having removed $\vSamp_j,\vSamp_{j+1},\ldots,\vSamp_{j+l-1}$) because $\vSamp_{j+l+1}-\vt_i \le s$. Finally, we removed $\vSamp_{j+k}$ because $\vt_{i+1}-\vt_i \le s$.
\end{proof}

This implies that, from a removed sample $\Samples[p]=\vSamp$, surrounded by the remaining samples $\vt_i < \vSamp < \vt_{i+1}$, we can perform only $k = \vSamp-\vt_i < s$ LF-steps until $j_k = \LF^{(k)}(j)$ satisfies $\SA[j_k]-1=\vt_i$ and thus it is stored in $\Samples[q]$ and not removed.

If we followed verbatim the modified backward search of the $\RIdx$, finding every $\SA[\vep_i]$, we would perform $\Oh(m \cdot s)$ steps on the $\SRIdx$. We now reduce this to $\Oh(m+s)$ steps by noting
that the only value we need is $\SA[\vep]=\SA[\vep_1]$. Further, we need to know $\SA[\vep_{i+1}]$ to compute $\SA[\vep_i]$ only in the easy case where $\BWT[\vep_{i+1}]=\vP[i]$ and so $\SA[\vep_i]=\SA[\vep_{i+1}]-1$. Otherwise, the value $\SA[\vep_i]$ is computed afresh. 

We then proceed as follows. We do not compute any value $\SA[\vep_i]$ during backward search; we only remember the last (i.e., smallest) value $i'$ of $i$ where the computation was not easy, that is, where $\BWT[\vep_{i'+1}]\not=\vP[i']$. Then, $\SA[\vep_1] = \SA[\vep_{i'}]-(i'-1)$ and we need to apply the procedure described above only once: we compute $\SA[j]$, where $j$ is the largest position in $[\vsp_{i'+1}..\vep_{i'+1}]$ where $\BWT[j]=P[i']$, and then $\SA[\vep_{i'}]=\SA[j]-1$.

Algorithm~\ref{alg:count} gives the complete pseudocode that solves Problem 1.
Note that, if $\vP$ does not occur in $\vT$ (i.e., $\occ=0$) we realize this after the 
$\Oh(m)$ backward steps because some $\vsp_i > \vep_i$, and thus we do not spend the $\Oh(s)$ extra steps. 

\begin{algorithm}[t]

\IO{Search pattern $\vP[1..\vm]$.}
    {Returns suffix array range $[\vsp,\vep]$ for $\vP$ and $\SA[\vep]$.}

$\vsp \leftarrow 1$; $\vep \leftarrow \vn+1$ \\
$i \leftarrow \vm$; $i' \leftarrow \vm+1$ \\
\While{$i \ge 1$ \KwAnd $\vsp \le \vep$}
   { $p \leftarrow \rank_1(\Start,\vep)$ \\     
     \If{{\em \Letter}$[p] \not= \vP[i]$}
        { $i' \leftarrow i$;     $p' \leftarrow p$ }
     $c \leftarrow \vP[i]$ \\
     $\vsp \leftarrow C[c]+\rank_c(\BWT,\vsp-1)+1$\\
     $\vep \leftarrow C[c]+\rank_c(\BWT,\vep)$
    }
\lIf{$\vsp > \vep$} {{\bf return} ``$\vP$ does not occur in $\vT$''}
\lIf{$i'=\vm+1$} {{\bf return} $[\vsp,\vep]$ and $\SA[\vep]=\SA[n]-\vm$ ($\SA[n]$ is stored)}
$c \leftarrow \vP[i']$ \\
$q \leftarrow \select_c(\Letter,\rank_c(\Letter,p'))$ (supported by the $\RLFMIdx$/$\RIdx$) \\
$j \leftarrow \select_1(\Start,q+1)-1$ \\
$k \leftarrow 0$ \\
\While{$(j < n$ \KwAnd {\em \Start}$[j+1]=0)$ \KwOr {\em \Removed}$[q] = 1$}
    { $j \leftarrow \LF(j)$ \\
      $q \leftarrow \rank_1(\Start,j)$ \\
      $k \leftarrow k+1$
    }
{\bf return} $[\vsp,\vep]$ and $\SA[\vep] = \Samples[q-\rank_1(\Removed,q)]+k+1-(i'-1)$

\caption{Counting pattern occurrences on the $\SRIdx$.}
\label{alg:count}
\end{algorithm}


\subsection{Solving Problem 2}

For Problem 2, finding $\SA[j-1]$ from $i=\SA[j]-1$, we first proceed as in Problem 1, from $j-1$.
We compute $j_k' = \LF^k(j-1)$ for
$k=0,\ldots,\vs-1$. If any of those $j_k'$ is the last symbol of its run (i.e., $j_k'=\vn$ or $\Start[j_k'+1]=1$), and the sample corresponding to
this run was not removed (i.e., $\Removed[q]=0$, with $q=\rank_1(\Start,j_k')$),
then we can obtain immediately $\SA[j_k'] = \Samples[q']+1$, where
$q'=q-\rank_1(\Removed,q)$, and thus $\SA[j-1]=\SA[j_k']+k$. 

Unlike in Problem 1, $\SA[j-1]$ is not necessarily an end of run, and therefore we are not guaranteed to find a solution for $0 \le k < s$. However, the following property shows that, if there were some end of runs $j_k'$, it is not possible that all were removed from $\Samples$.

\begin{lemma} \label{lem:noremoved}
If there are no remaining samples in $\SA[j-1]-s,\ldots,\SA[j-1]-1$, then no sample was removed between $\SA[j-1]-1$ and its preceding remaining sample.
\end{lemma}
\begin{proof}
Let $\vt_i < \SA[j-1]-1 < \vt_{i+1}$ be the samples surrounding $\SA[j-1]-1$, so the remaining sample preceding $\SA[j-1]-1$ is $\vt_i$. Since $\vt_i < \SA[j-1]-s$, it follows that $\vt_{i+1}-\vt_i > s$ and thus, by Lemma~\ref{lem:distance-s}, no samples were removed between $\vt_i$ and $\vt_{i+1}$.
\end{proof}

This means that, if the process above fails to find an answer, then 
we can directly use Eq.~(\ref{eq:rindex}), as we prove next.

\begin{lemma}
If there are no remaining samples in $\SA[j-1]-s,\ldots,\SA[j-1]-1$, then subsampling removed no $1$s in {\em $\First$} between positions $i=\SA[j]-1$ and {\em $\pred(\First,i)$}.
\end{lemma}
\begin{proof}
Let $\vt_i < \SA[j-1]-1 < \vt_{i+1}$ be the samples surrounding $\SA[j-1]-1$, and
$k=\SA[j-1]-1-\vt_i$. Lemma~\ref{lem:noremoved} implies that no sample existed between $\SA[j-1]-1$ and $\SA[j-1]-k=t_i+1$, and there exists one at $t_i$. Consequently, no $1$ existed in $\First$ between positions $\SA[j]-1$ and $\SA[j]-k$, and there exists one in $\SA[j]-1-k$. Indeed, $\pred(\First,i)=\SA[j]-1-k$.
\end{proof}

A final twist, which does not change the worst-case complexity but improves performance in practice, is to reuse work among successive occurrences. Let $\BWT[sm..em]$ be a maximal run inside $\BWT[\vsp..\vep]$. For every $sm \le j \le em$, the first LF-step will lead us to $\LF(j) = \LF(sm)+(j-sm)$; therefore we can obtain them all with only one computation of $\LF$.
Therefore, instead of finding $\SA[\vsp],\ldots,\SA[\vep]$ one by one, we report $\SA[\vep]$ (which we know) and cut $\BWT[\vsp..\vep-1]$ into maximal runs using bitvector $\Start$. Then, for each maximal run $\BWT[sm..em]$, if the end of run $\BWT[em]$ is sampled, we report its position and continue recursively reporting $\SA[\LF(sm)..\LF(sm)+(em-sm)-1]$; otherwise we continue recursively reporting $\SA[\LF(sm)..\LF(sm)+(em-sm)]$. Note that we must add $k$ to the results reported at level $k$ of the recursion. By Lemma~\ref{lem:distance-s}, every end of run found in the way has been reported before level $k=s$. When $k=s$, then, we use Eq.~(\ref{eq:rindex}) to obtain $\SA[em],\ldots,\SA[sm]$ consecutively from $\SA[em+1]$, which must have been reported because it is $\vep$ or was an end of run at some level of the recursion. 

Algorithm~\ref{alg:locate} gives the complete procedure to solve Problem 2.

\begin{algorithm}[t]

\IO{Global array $Res[1..\occ]$ of results, range $[\vsp,\vep]$ to report, $\SA[\vep]$.}
    {Fills $Res[i] = \SA[sp-1+i]$ for all $1 \le i \le \occ$.}

$Res[\vep-\vsp+1] \leftarrow \SA[\vep]$ (known from backward search) \\
\lIf{$\vsp<\vep$} {$locate(1,\vep-\vsp,0)$}

\medskip

\Fn{$locate(sm,em,k)$}{
 \If{$k = \vs$}
    { \For{$im = em,\ldots,sm$}
         { $i \leftarrow Res[im+1]-1$ \\
         $Res[im] \leftarrow \phi(i)$ (Eq.~(\ref{eq:rindex}))
         }
     }
  \Else
     { \If{{\em \Start}$[\vsp+em]=1$}
          { $q \leftarrow \rank_1(\Start,\vsp-1+em)$ \\
          \If{{\em \Removed}$[q]=0$}
             { $Res[em] \leftarrow \Samples[q-\rank_1(\Removed,q)]+1+k$ \\
               $em \leftarrow em-1$
             }
          }
       \While{$sm \le em$}
          { $q \leftarrow \rank_1(\Start,\vsp-1+sm)$ \\
          $im \leftarrow \select_1(\Start,q+1)$ \\
          \lIf{$im-1>em$} {$im \leftarrow em+1$}
          $locate(sm,im-1,k+1)$ \\
          $sm \leftarrow im$
          }
     }
}
\caption{Locating pattern occurrences on the $\SRIdx$.}
\label{alg:locate}
\end{algorithm}

\subsection{The basic index, $\SRIdx_0$}

We have just described our most space-efficient index, which we call $\SRIdx_0$. Its space and time complexity is established in the next theorem.

\begin{theorem} \label{thm:srindex}
The $\SRIdx_0$ uses $r \cdot ((1+\epsilon)\lg(n/r)+\lg\sigma+\Oh(1)) + \min(r,2\lceil n/(s+1)\rceil)\cdot 2\lg n$ bits of space, for any constant $\epsilon>0$, and finds all the $\occ$ occurrences of $\vP[1..\vm]$ in $\vT$ in time $\Oh((m+s\cdot occ)\log(\sigma+n/r))$.
\end{theorem}
\begin{proof}
The space is the sum of the counting structures of the $\RIdx$ and our modified locating structures, according to Lemma~\ref{lem:space}. The space of bitvector $\Removed$ is $\Oh(r)$ bits, which is accounted for in the formula.

As for the time, we have seen that the modified backward search requires $\Oh(\vm)$ steps if $\occ=0$ and $\Oh(\vm+\vs)$ otherwise (Problem 1). Each occurrence is then located in $\Oh(\vs)$ steps (Problem 2). In total, we complete the search with $\Oh(\vm+\vs\cdot \occ)$ steps.

Each step involves $\Oh(\log(\sigma+n/r))$ time in the basic $\RIdx$ implementation, including Eq.~(\ref{eq:rindex}). Our index includes additional $\rank$s on $\Start$ and other constant-time operations, which are all in $\Oh(\log(n/r))$. Since the $\First$ now has $\Oh(\min(r,n/s))$ $1$s, however, operation $\rank_1$ on it takes time $\Oh(\log (n/\min(r,n/s))) = 
\Oh(\log\max(n/r,s)) = \Oh(\log(n/r+s))$. Yet, this $\rank$ is computed only once per occurrence reported, when using Eq.~(\ref{eq:rindex}), so the total time per occurrence is still $\Oh(\log(n/r+s)+s \cdot \log(\sigma+n/r)) = \Oh(s \cdot \log(\sigma+n/r))$.
\end{proof}

Note that, in asymptotic terms, the $\SRIdx$ is never worse than the $\RLFMIdx$ with the same value of $s$ and, with $s=1$, it boils down to the $\RIdx$.
Using predecessor data structures of the same asymptotic space of our lighter sparse bitvectors, the logarithmic times can be reduced to loglogarithmic \cite{GagieNP20:FFS}, but our focus is on low practical space usage.

Note also that this theorem can be obtained by simply choosing the smallest between the $\RIdx$ and the $\RLFMIdx$. In practice, however, the $\SRIdx$ performs much better than both extremes, providing a smooth transition that retains sparsely indexed areas of $\vT$ while removing redundancy in oversampled areas. This will be demonstrated in Section~\ref{sec:result}.

\subsection{A faster and larger index, $\SRIdx_1$}

The $\SRIdx_0$ guarantees locating time proportional to $\vs$ and uses almost no extra
space. On the other hand, on Problem 2 it performs up to $s$ LF-steps for {\em every} 
occurrence, even when this turns out to be useless. 
The variant $\SRIdx_1$ adds a new component, also small, to speed up some cases:

\begin{description}
\item[{$\Valid$}:] a bitvector storing one bit per (remaining) sample in text
order, so that $\Valid[q]=0$ iff there were removed samples between the $q$th 
and the $(q+1)$th $1$s of $\First$. 
\end{description}

With this bitvector, if we have $i=\SA[j]-1$ and $\Valid[\rank_1(\First,i)]=1$, we 
know that there were no removed samples between $i$ and $\pred(\First,i)$ (even if they are  less than $s$ positions apart). In this case we can skip the computation of $\LF^k(j-1)$ of $\SRIdx_0$, and 
directly use Eq.~(\ref{eq:rindex}). Otherwise, we must proceed exactly as
in $\SRIdx_0$ (where it is still possible that we compute all the LF-steps
unnecessarily).
More precisely, this can be tested for every value between $sm$ and $em$ so as to report some further cells before recursing on the remaining ones, in lines 14--19 of Algorithm~\ref{alg:locate}. 

The space and worst-case complexities of Theorem~\ref{thm:srindex} are preserved in $\SRIdx_1$.

\subsection{Even faster and larger, $\SRIdx_2$}

Our final variant, $\SRIdx_2$, adds a second and significantly larger structure:

\begin{description}
\item[{$\ValidArea$}:] an array whose cells are associated with the $0$s in 
$\Valid$. If $\Valid[q]=0$, then $d = \ValidArea[q-\rank_1(\Valid,q)]$ is 
the distance from the $q$th $1$ in $\First$ to the next removed sample. Each
entry in $\ValidArea$ requires $\lceil \lg \vs \rceil$ bits, because
removed samples must be at distance less than $\vs$ from their preceding sample,
by Lemma~\ref{lem:distance-s}.
\end{description}

If $\Valid[\rank_1(\First,i)]=0$, then there was a removed sample at $\pred(\First,i)+d$, but not before. So, if $i < \pred(\First,i)+d$, we can still use Eq.~(\ref{eq:rindex}); otherwise we must compute the LF-steps $\LF^k(j-1)$ and we are guaranteed to succeed in less than $s$ steps. This improves performance considerably in practice, though the worst-case time complexity stays as in Theorem~\ref{thm:srindex} and the space increases by at most $\vr\lg \vs$ bits.

\section{Experimental Results} \label{sec:result}

We implemented the $\SRIdx$ in C++14, on top of the SDSL library\footnote{From {\tt https://github.com/simongog/sdsl-lite}.}, and made it available at {\tt https://github.com/duscob/sr-index}.

We benchmarked the $\SRIdx$ against available implementations for the $\RIdx$, the $\RLFMIdx$, and several other indexes for repetitive text collections.

Our experiments ran on a hardware with two Intel(R) Xeon(R) CPU E5-2407 processors at $2.40$ GHz and $250$ GB RAM.
The operating system was Debian Linux kernel \texttt{4.9.0-14-amd64}.
We compiled with full optimization and no multithreading.

Our reported times are the average user time over 1000 searches for patterns of length $m=10$ obtained at random from the texts. We give space in bits per symbol (bps) and times in microseconds per occurrence ($\mu$s/occ). Indexes that could not be built on some collection, or that are out of scale in space or time, are omitted in the corresponding plots.

\subsection{Tested indexes}

We included the following indexes in our benchmark; their space decrease as $s$ grows:

\begin{description}
\item[$\SRIdx$:] Our index, including the three variants, with sampling values $s=4, 8, 16, 32, 64$.
\item[$\RIdx$:] The $\RIdx$ implementation we build on.\footnote{From {\tt https://github.com/nicolaprezza/r-index}.}
\item[\textsf{rlcsa}:] An implementation of the run-length CSA \cite{MakinenNSV10:SRH}, which outperforms the actual $\RLFMIdx$ implementation.\footnote{From {\tt https://github.com/adamnovak/rlcsa}.} We use text sampling values $s=\vsTColl / \vnBWTRuns \times f / 8$, with $f=8, 10, 12, 14, 16$.
\item[\textsf{csa}:] An implementation of the CSA \cite{Sad03}, which outperforms in practice the $\FMIdx$ \cite{FerraginaM05:ICT,FerraginaMMN07:CRS}. This index, obtained from SDSL, acts as a control baseline that is not designed for repetitive collections.
 We use text sampling parameter $s=16, 32, 64, 128$.
\item[$\GIdx$:] The best grammar-based index implementation we are aware of \cite{CNP21}.\footnote{From {\tt https://github.com/apachecom/grammar\_improved\_index}.} We use Patricia trees sampling values $s=4, 16, 64$.
\item[$\LZIdx$ and $\LZEndIdx$:] Two variants of the Lempel-Ziv based index \cite{KreftN13:CIR}.\footnote{From {\tt https://github.com/migumar2/uiHRDC}.}
\item[$\HybridIdx$:] A hybrid between a Lempel-Ziv and a $\BWT$-based index \cite{FKP18}.\footnote{From {\tt https://github.com/hferrada/HydridSelfIndex}.} We build it with parameters $M = 8,16$, the best for this case. 
\end{description}

\subsection{Collections}

We benchmark various repetitive text collections; Table~\ref{tab:texts} gives some basic measures on them.

\begin{description}
\item[PizzaChili:] A generic collection of real-life texts of various sorts and repetitiveness levels, which we use to obtain a general idea of how the indexes compare. We use 4 collections of microorganism genomes (\textsf{influenza}, \textsf{cere}, \textsf{para}, and \textsf{escherichia}) and 4 versioned document collections (the English version of \textsf{einstein}, \textsf{kernel}, \textsf{worldleaders}, \textsf{coreutils}).\footnote{From {\tt http://pizzachili.dcc.uchile.cl/repcorpus/real}.}
\item[Synthetic DNA:] A 100KB DNA text from PizzaChili, replicated $1{,}000$ times and each copied symbol mutated with a probability from $0.001$ (\textsf{DNA-001}, analogous to human assembled genomes) to $0.03$ (\textsf{DNA-030}, analogous to sequence reads). We use this collection to study how the indexes evolve as repetitiveness decreases.
\item[Real DNA:] Some real DNA collections to study other aspects:
\begin{description}
\item[\textsf{HLA}:] A dataset with copies of the short arm (p arm) of human chromosome 6 \cite{RBGCFM19}.\footnote{From {\tt ftp://ftp.ebi.ac.uk/pub/databases/ipd/imgt/hla/fasta/hla\_gen.fasta}.}  This arm contains about 60 million base pairs (Mbp) and it  includes the 3 Mbp HLA region.  That region is known to be highly variable, so the $\RIdx$ sampling should be sparse for most of the arm and oversample the HLA region.
\item[\textsf{Chr19} and \textsf{Salmonella}:] Human and bacterial assembled genome collections, respectively, of a few billion base pairs. We include them to study how the indexes behave on more massive data. \textsf{Chr19} is the set of 50 human chromosome 19 genomes taken from the 1000 Genomes Project~\cite{1000genomes}, whereas \textsf{Salmonella} is the set of 815 Salmonella genomes from the GenomeTrakr project~\cite{stevens2017public}.
\item[\textsf{Reads}:] A large collection of sequence reads, which tend to be considerably less repetitive than assembled genomes.\footnote{From {\tt https://trace.ncbi.nlm.nih.gov/Traces/sra/?run=ERR008613}.} We include this collection to study the behavior of the indexes on a popular kind of bioinformatic collection with mild repetitiveness. In \textsf{Reads} the sequencing errors have been corrected, and thus its $n/r \approx 9$ is higher than the $n/r \approx 4$ reported on crude reads \cite{DN21}.
\end{description}
\end{description}

\begin{table}[t]
\begin{center}
\begin{tabular}{l|r|r || l|r|r}
Collection & Size & $n/r$ & Collection & Size & $n/r$ \\
\hline
\textsf{influenza} & 147.6 & 51.2 &
\textsf{DNA-001} & 100.0 & 142.4 \\ 
\textsf{cere} & 439.9 & 39.9 & 
\textsf{DNA-003} & 100.0 & 58.3 \\ 
\textsf{para} & 409.4 & 27.4 & 
\textsf{DNA-010} & 100.0 & 26.0 \\ 
\textsf{escherichia} & 107.5 & 7.5 &
\textsf{DNA-030} & 100.0 & 11.6 \\
\hline
\textsf{einstein} & 447.7 & 1611.2 &
\textsf{HLA}     & 53.7 & 161.4 \\ 
\textsf{kernel} & 238.0 & 92.4 &
\textsf{Chr19} & 2{,}819.3 & 89.2 \\ 
\textsf{worldleaders} & 44.7 & 81.9 &
\textsf{Salmonella}     & 3{,}840.5 & 43.9 \\ 
\textsf{coreutils} & 195.8 & 43.8 &
\textsf{Reads}     & 2{,}565.5 & 8.9 \\
\end{tabular}
\end{center}
\caption{Basic characteristics of the repetitive texts used in our benchmark. Size is given in MB.}
\label{tab:texts}
\end{table}

\subsection{Results}

Figures~\ref{fig:exp1} and \ref{fig:exp2} show the space taken by all the indexes and their search time.

\begin{figure}[p]
\includegraphics[width=0.49\textwidth]{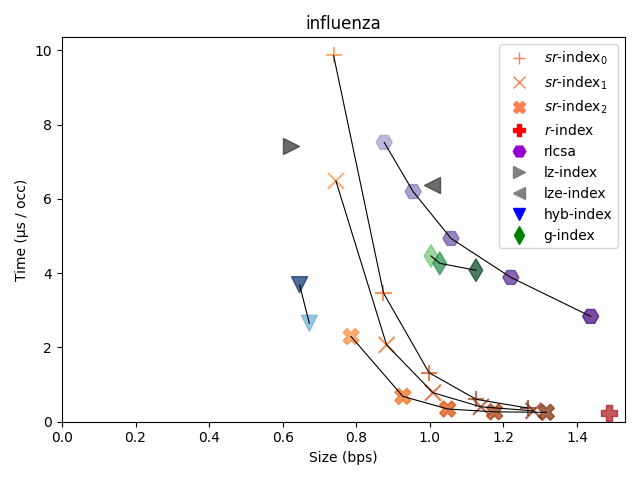}
\includegraphics[width=0.49\textwidth]{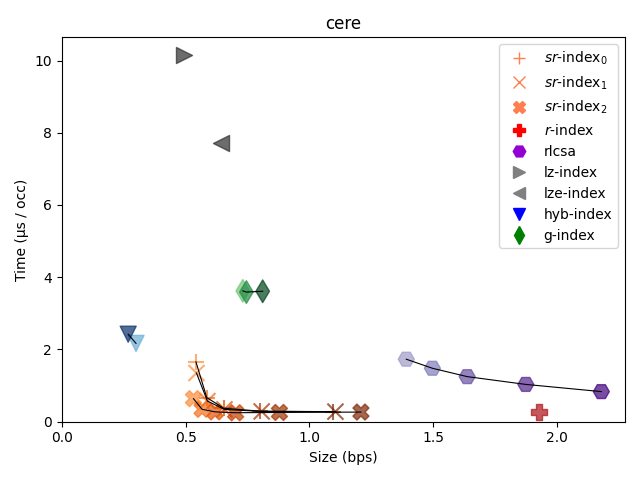}
\includegraphics[width=0.49\textwidth]{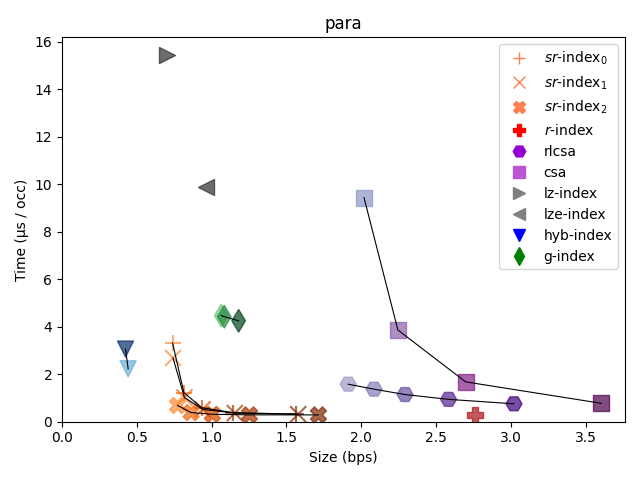}
\includegraphics[width=0.49\textwidth]{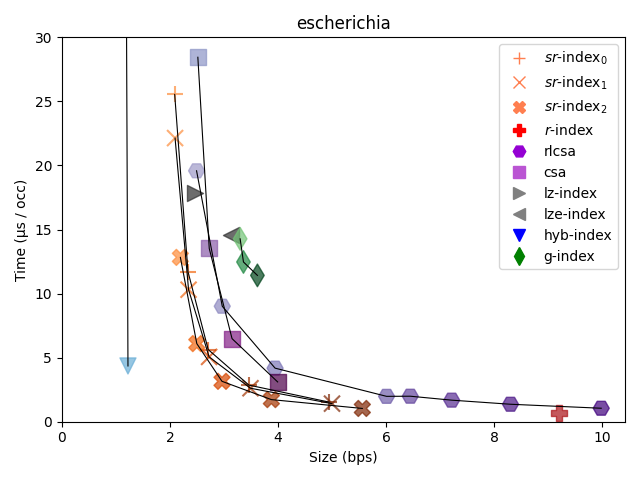}
\includegraphics[width=0.49\textwidth]{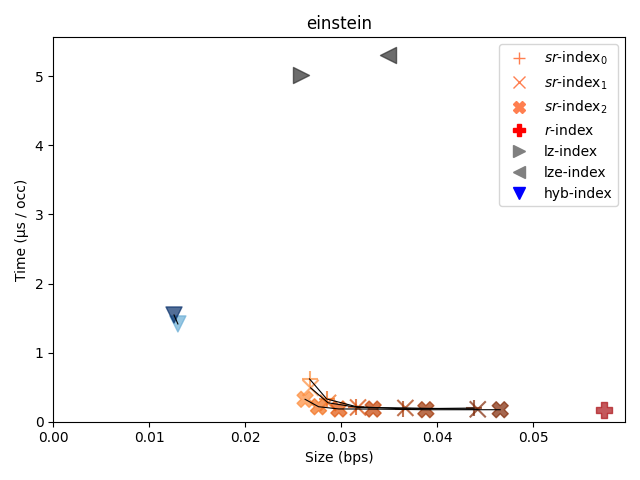}
\includegraphics[width=0.49\textwidth]{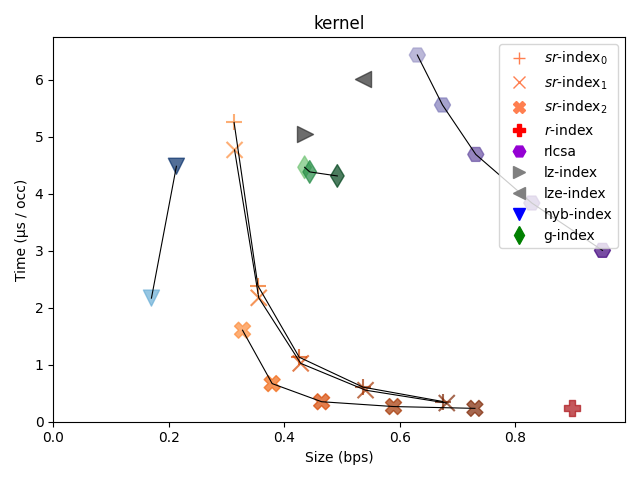}
\includegraphics[width=0.49\textwidth]{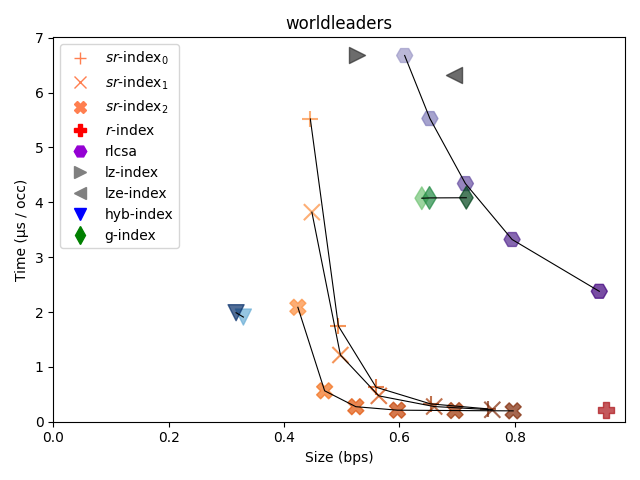}
\includegraphics[width=0.49\textwidth]{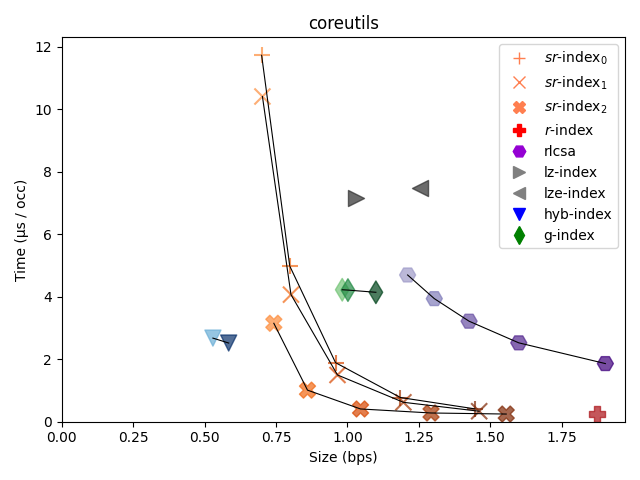}
\caption{Space-time tradeoffs for the PizzaChili collections.}
\label{fig:exp1}
\end{figure}

\begin{figure}[p]
\includegraphics[width=0.49\textwidth]{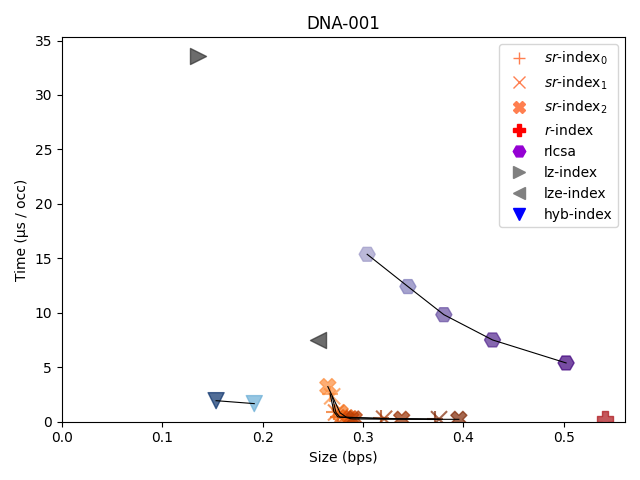}
\includegraphics[width=0.49\textwidth]{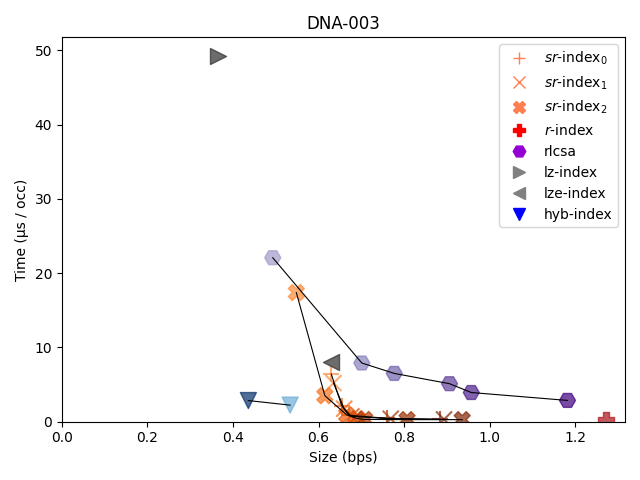}
\includegraphics[width=0.49\textwidth]{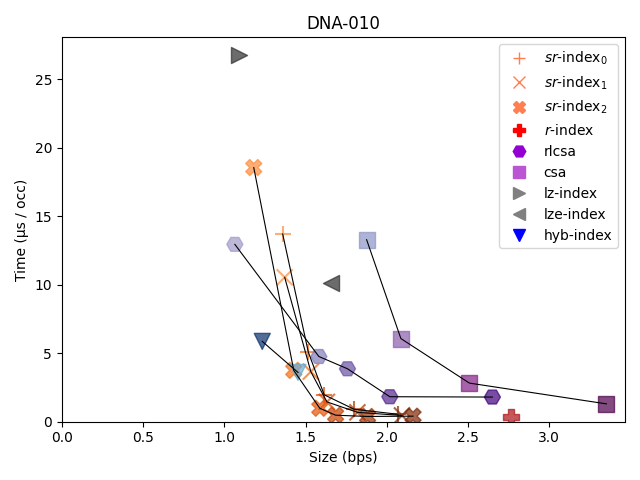}
\includegraphics[width=0.49\textwidth]{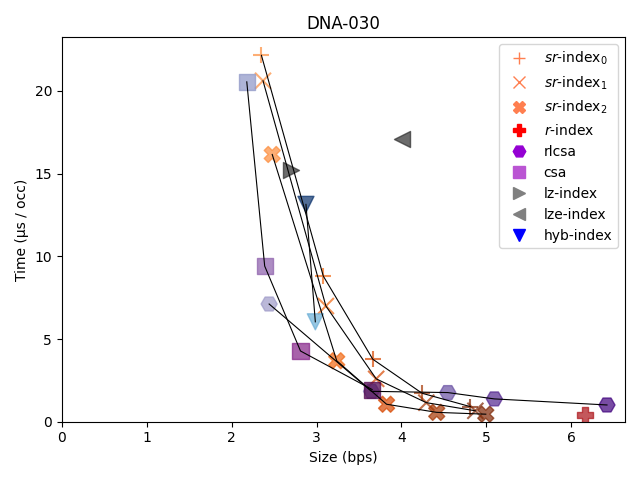}
\includegraphics[width=0.49\textwidth]{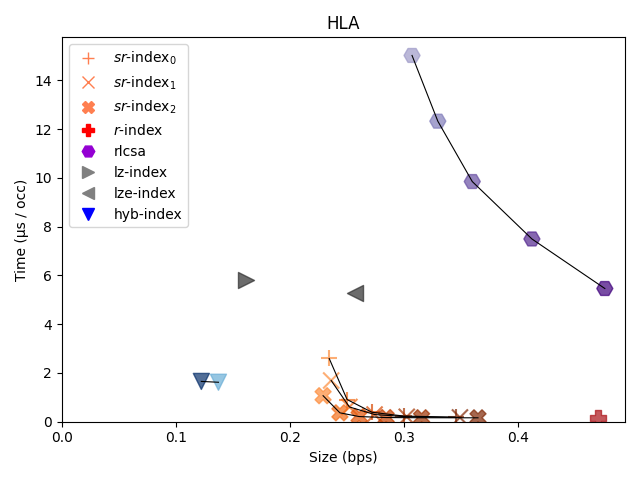}
\includegraphics[width=0.49\textwidth]{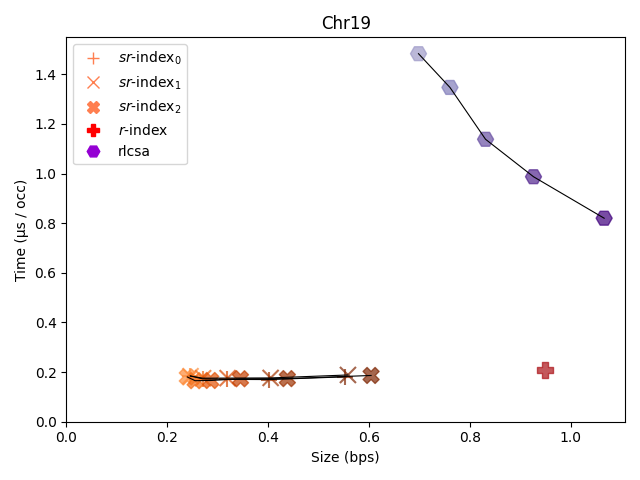}
\includegraphics[width=0.49\textwidth]{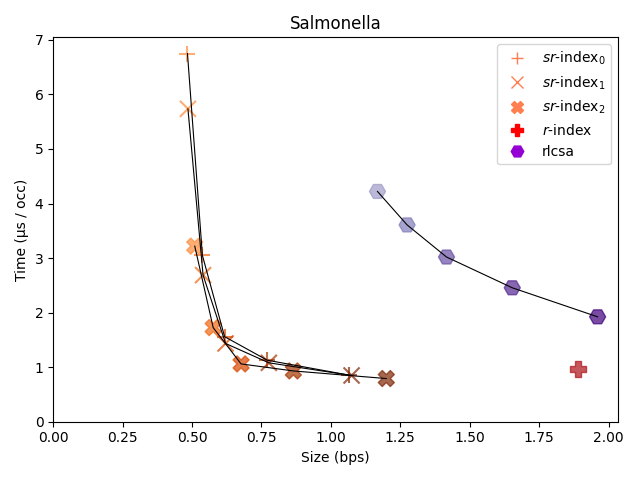}
\includegraphics[width=0.49\textwidth]{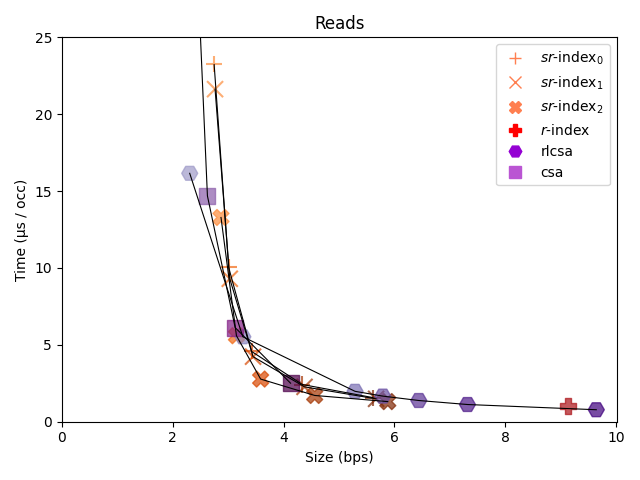}
\caption{Space-time tradeoffs for the synthetic and real DNA datasets.}
\label{fig:exp2}
\end{figure}

A first conclusion is that $\SRIdx_2$ always dominates $\SRIdx_0$ and $\SRIdx_1$, so we will refer to it simply as $\SRIdx$ from now on. The plots show that the extra information we associate to the samples makes a modest difference in space, while time improves considerably.
This $\SRIdx$ can be almost as fast as the $\RIdx$, and an order of magnitude faster than all the others, while using $1.5$--$4.0$ less space than the $\RIdx$. Therefore, as promised, we are able to remove a significant degree of redundancy in the $\RIdx$ without affecting its outstanding time performance.

In all the PizzaChili collections, the $\SRIdx$ dominates almost every other index, outperforming them both in time and space. The only other index on the Pareto curve is the $\HybridIdx$, which can use as little as a half of the space of the sweet spot of the $\SRIdx$, but still at the price of being an order of magnitude slower. This holds even on \textsf{escherichia}, where $n/r$ is well below $10$, and both the \textsf{rlcsa} and the \textsf{csa} become closer to the $\SRIdx$.

In general, in all the collections with sufficient repetitiveness, say $n/r$ over 25, the $\SRIdx$ sharply dominates as described.
As repetitiveness decreases, with $n/r$ reaching around 10, the \textsf{rlcsa} and the \textsf{csa} approach the $\SRIdx$ and outperform every other repetitiveness-aware index, as expected. This happens on \textsf{escherichia} (as mentioned) and \textsf{Reads} (where the $\SRIdx$, the \textsf{rlcsa}, and the \textsf{csa} behave similarly). This is also the case on the least repetitive synthetic DNA collection, \textsf{DNA-030}, where the mutation rate reaches 3\%. In this collection, the repetitiveness-unaware \textsf{csa} largely dominates all the space-time map.

We expected the $\SRIdx$ to have a bigger advantage over the $\RIdx$ on the \textsf{HLA} dataset because its oversampling is concentrated, but the results are similar to those on randomly mutated DNA with about the same $n/r$ value (\textsf{DNA-001}). In general, the bps used by the $\SRIdx$ can be roughly predicted from $n/r$; for example the sweet spot often uses around $40r$ total bits, although it takes $20r$--$30r$ bits in some cases. The $\RIdx$ uses $70r$--$90r$ bits.

The bigger collections (\textsf{Chr19}, \textsf{Salmonella}, \textsf{Reads}), on which we could build the $\BWT$-related indexes only, show that the same observed trends scale to gigabyte-sized collections of various repetitiveness levels.

\section{Conclusions} \label{sec:conclusion}

We have introduced the $\SRIdx$, an $\RIdx$ variant that solves the problem of its relatively bloated space while retaining its high search performance. The $\SRIdx$ is orders of magnitude faster than the other repetitiveness-aware indexes, while outperforming most of them in space as well. It matches the time performance of the $\RIdx$ while using $1.5$--$4.0$ less space.

Unlike the $\RIdx$, the $\SRIdx$ uses little space even in milder repetitiveness scenarios, which makes it usable in a wider range of bioinformatic applications.
For example, it uses 0.25--0.60 bits per symbol (bps) while reporting each occurrence within a microsecond on gigabyte-sized human and bacterial  genomes, where the original $\RIdx$ uses 0.95--1.90 bps. In general, the $\SRIdx$ outperforms classic compressed indexes on collections with repetitiveness levels $n/r$ over as little as $7$ in some cases, though in general it is reached by repetitiveness-unaware indexes when $n/r$ approaches $10$, which is equivalent to a DNA mutation rate around 3\%.

Compared to the $\RLFMIdx$, which for pattern searching is dominated by the $\SRIdx$, the former can use its regular text sampling to compute any entry of the suffix array or its inverse in time proportional to the sampling step $\vs$. Obtaining an analogous result on the $\SRIdx$, for example to implement compressed suffix trees, is still a challenge. Other proposals for accessing the suffix array faster than the $\RLFMIdx$ \cite{GNF14,PZ20} illustrate this difficulty: they require even more space than the $\RIdx$.

%% file: ms.bbl
\begin{thebibliography}{10}

\bibitem{BCGHMNR21}
Christina Boucher, Ondrej Cvacho, Travis Gagie, Jan Holub, Giovanni Manzini,
  Gonzalo Navarro, and Massimiliano Rossi.
\newblock {PFP} compressed suffix trees.
\newblock In {\em Proc. 23rd Workshop on Algorithm Engineering and Experiments
  (ALENEX)}, pages 60--72, 2021.

\bibitem{BGKLMM19}
Christina Boucher, Travis Gagie, Alan Kuhnle, Ben Langmead, Giovanni Manzini,
  and Taher Mun.
\newblock Prefix-free parsing for building big {BWTs}.
\newblock {\em Algorithms for Molecular Biology}, 14(1):13:1--13:15, 2019.

\bibitem{BurrowsW94:BSL}
Michael Burrows and David~J. Wheeler.
\newblock {A block-sorting lossless data compression algorithm}.
\newblock Technical Report 124, Digital Equipment Corporation, 1994.

\bibitem{Cla96}
D.~R. Clark.
\newblock {\em Compact {PAT} Trees}.
\newblock PhD thesis, University of Waterloo, Canada, 1996.

\bibitem{CNP21}
Francisco Claude, Gonzalo Navarro, and Alejandro Pacheco.
\newblock Grammar-compressed indexes with logarithmic search time.
\newblock {\em Journal of Computer and System Sciences}, 118:53--74, 2021.

\bibitem{DN21}
Diego D{\'i}az-Dom{\'i}nguez and Gonzalo Navarro.
\newblock A grammar compressor for collections of reads with applications to
  the construction of the {BWT}.
\newblock In {\em Proc. 31st Data Compression Conference (DCC)}, 2021.
\newblock To appear.

\bibitem{FKP18}
H{\'{e}}ctor Ferrada, Dominik Kempa, and Simon~J. Puglisi.
\newblock Hybrid indexing revisited.
\newblock In {\em Proc. 20th Workshop on Algorithm Engineering and Experiments
  (ALENEX)}, pages 1--8, 2018.

\bibitem{FerraginaM05:ICT}
Paolo Ferragina and Giovanni Manzini.
\newblock {Indexing Compressed Text}.
\newblock {\em Journal of the ACM}, 52(4):552--581, 2005.

\bibitem{FerraginaMMN07:CRS}
Paolo Ferragina, Giovanni Manzini, Veli M{\"{a}}kinen, and Gonzalo Navarro.
\newblock {Compressed Representations of Sequences and Full-text Indexes}.
\newblock {\em ACM Transactions on Algorithms}, 3(2), 2007.

\bibitem{GNencyc18}
T.~Gagie and G.~Navarro.
\newblock {\em Compressed Indexes for Repetitive Textual Datasets}.
\newblock Springer, 2019.

\bibitem{GagieNP20:FFS}
Travis Gagie, Gonzalo Navarro, and Nicola Prezza.
\newblock Fully-functional suffix trees and optimal text searching in
  {BWT}-runs bounded space.
\newblock {\em Journal of the ACM}, 67(1):article 2, 2020.

\bibitem{GNF14}
Rodrigo Gonz{\'a}lez, Gonzalo Navarro, and H{\'e}ctor Ferrada.
\newblock Locally compressed suffix arrays.
\newblock {\em ACM Journal of Experimental Algorithmics}, 19(1):article 1,
  2014.

\bibitem{KK19}
Dominik Kempa and Tomasz Kociumaka.
\newblock Resolution of the {B}urrows-{W}heeler transform conjecture, 2019.
\newblock To appear in {\em FOCS 2020}.
\newblock \href {http://arxiv.org/abs/1910.10631} {\path{arXiv:1910.10631}}.

\bibitem{KY00}
John~C. Kieffer and En-Hui Yang.
\newblock Grammar-based codes: {A} new class of universal lossless source
  codes.
\newblock {\em IEEE Transactions on Information Theory}, 46(3):737--754, 2000.

\bibitem{KreftN13:CIR}
Sebastian Kreft and Gonzalo Navarro.
\newblock {On compressing and indexing repetitive sequences}.
\newblock {\em Theoretical Computer Science}, 483:115--133, 2013.

\bibitem{LZ76}
Abraham Lempel and Jacob Ziv.
\newblock On the complexity of finite sequences.
\newblock {\em IEEE Transactions on Information Theory}, 22(1):75--81, 1976.

\bibitem{MBCT15}
Veli M{\"a}kinen, Djamal Belazzougui, Fabio Cunial, and Alexandru~I. Tomescu.
\newblock {\em Genome-Scale Algorithm Design}.
\newblock Cambridge University Press, 2015.

\bibitem{MakinenN05:SSA}
Veli M{\"a}kinen and Gonzalo Navarro.
\newblock Succinct suffix arrays based on run-length encoding.
\newblock {\em Nordic Journal of Computing}, 12(1):40--66, 2005.

\bibitem{MakinenNSV10:SRH}
Veli M{\"a}kinen, Gonzalo Navarro, Jouni Sir{\'e}n, and Niko V{\"a}lim{\"a}ki.
\newblock Storage and retrieval of highly repetitive sequence collections.
\newblock {\em Journal of Computational Biology}, 17(3):281--308, 2010.

\bibitem{ManberM93:SAN}
Udi Manber and Gene Myers.
\newblock Suffix arrays: a new method for on-line string searches.
\newblock {\em SIAM Journal on Computing}, 22(5):935--948, 1993.

\bibitem{Nav20}
Gonzalo Navarro.
\newblock Indexing highly repetitive string collections.
\newblock {\em CoRR}, 2004.02781, 2020.
\newblock To appear in {\em ACM Computing Surveys}.

\bibitem{NavarroM07:CFT}
Gonzalo Navarro and Veli M{\"a}kinen.
\newblock Compressed full-text indexes.
\newblock {\em ACM Computing Surveys}, 39(1):article 2, 2007.

\bibitem{NS19}
Gonzalo Navarro and V{\'i}ctor Sep{\'u}lveda.
\newblock Practical indexing of repetitive collections using {R}elative
  {L}empel-{Z}iv.
\newblock In {\em Proc. 29th Data Compression Conference (DCC)}, pages
  201--210, 2019.

\bibitem{NT20}
Takaaki Nishimoto and Yasuo Tabei.
\newblock Faster queries on {BWT}-runs compressed indexes.
\newblock {\em CoRR}, 2006.05104, 2020.

\bibitem{OS07}
Daisuke Okanohara and Kunihiko Sadakane.
\newblock Practical entropy-compressed rank/select dictionary.
\newblock In {\em Proc. 9th Workshop on Algorithm Engineering and Experiments
  (ALENEX)}, pages 60--70, 2007.

\bibitem{PZ20}
Simon~J. Puglisi and Bella Zhukova.
\newblock Relative {Lempel-Ziv} compression of suffix arrays.
\newblock In {\em Proc. 27th International Symposium on String Processing and
  Information Retrieval (SPIRE)}, pages 89--96, 2020.

\bibitem{RBGCFM19}
James Robinson, Dominic~J. Barker, Xenia Georgiou, Michael~A. Cooper, Paul
  Flicek, and Steven G.~E. Marsh.
\newblock {IPD-IMGT/HLA Database}.
\newblock {\em Nucleic Acids Research}, 48(D1):D948--D955, 10 2019.
\newblock \href {https://doi.org/10.1093/nar/gkz950}
  {\path{doi:10.1093/nar/gkz950}}.

\bibitem{Sad03}
Kunihiko Sadakane.
\newblock New text indexing functionalities of the compressed suffix arrays.
\newblock {\em Journal of Algorithms}, 48(2):294--313, 2003.

\bibitem{stevens2017public}
Eric~L. Stevens, Ruth Timme, Eric~W. Brown, Marc~W. Allard, Errol Strain, Kelly
  Bunning, and Steven Musser.
\newblock The public health impact of a publically available, environmental
  database of microbial genomes.
\newblock {\em Frontiers in Microbiology}, 8:808, 2017.

\bibitem{1000genomes}
{The 1000 Genomes Project Consortium}.
\newblock A global reference for human genetic variation.
\newblock {\em Nature}, 526:68--74, 2015.

\end{thebibliography}
